\documentclass[10pt,a4]{article}

\usepackage{multicol}
\usepackage{graphicx}
\usepackage{amsmath, amssymb, amsthm}
\usepackage{times}
\usepackage[colorlinks=true, linkcolor = blue]{hyperref}
\usepackage{sidecap}
\usepackage{subfig}
\usepackage{microtype}
\usepackage{tikz}

\def\ve{\varepsilon}

\def\L{\Lambda}
\def\T{\Theta}
\def\l{\lambda}

\def\beq{\begin {equation}}
\def\eeq{\end {equation}}
\def\beqar{\begin {eqnarray*}}
\def\eeqar{\end {eqnarray*}}

\newcommand{\expect}{\mathbb{E}}
\newcommand{\real}{\mathbb{R}}
\newcommand{\prob}{\mathbb{P}}

\newcommand{\calG}{\mathcal{G}}

\newcommand{\calM}{\mathcal{M}}
\newcommand{\calP}{\mathcal{P}}
\newcommand{\calS}{\mathcal{S}}
\newcommand{\calC}{\mathcal{C}}
\newcommand{\calT}{\mathcal{T}}

\newcommand{\ie}{\emph{i.e.}}

\newtheorem{thm}{Theorem}
\newtheorem*{uthm}{Theorem}

\newtheorem{mydef}{Definition}

\newtheorem{conj}{Conjecture}

\newtheorem{cor}[thm]{Corollary}
\newtheorem{lem}[thm]{Lemma}

\author{
Pranav Dandekar\footnote{Department of Management Science \& Engineering, Stanford University. Email: \href{mailto:ppd@stanford.edu}{ppd@stanford.edu}.
 }
\and 
Ashish Goel\footnote{Departments of Management Science \& Engineering and (by courtesy) Computer Science, Stanford University. Email: \href{mailto:ashishg@stanford.edu}{ashishg@stanford.edu}. 
 }
\and 
Ramesh Govindan\footnote{Computer Science Department, University of Southern California. Email: \href{mailto:ramesh@usc.edu}{ramesh@usc.edu}. 
}
\and
Ian Post\footnote{Computer Science Department, Stanford University. Email: \href{mailto:itp@stanford.edu}{itp@stanford.edu}.
 }
}

\begin{document}
\title{Liquidity in Credit Networks: A Little Trust Goes a Long Way}

\vspace{-0.5in}
\date{}
\maketitle

\begin{abstract}
Credit networks represent a way of modeling trust between entities in a network. Nodes in the network print their own currency and trust each other for a certain amount of each other's currency. This allows the network to serve as a decentralized payment infrastructure---arbitrary payments can be routed through the network by passing IOUs between trusting nodes in their respective currencies---and obviates the need for a common currency. Nodes can repeatedly transact with each other and pay for the transaction using trusted currency. A natural question to ask in this setting is: how long can the network sustain liquidity, \ie, how long can the network support the routing of payments before credit dries up? We answer this question in terms of the long term failure probability of transactions for various network topologies and credit values.

We prove that the transaction failure probability is independent of the path along which transactions are routed. We show that under symmetric transaction rates, the transaction failure probability in a number of well-known graph families goes to zero as the size,  density or credit capacity of the network increases. We also show via simulations that even networks of small size and credit capacity can route transactions with high probability if they are well-connected. Further, we characterize a centralized currency system as a special type of a star network (one where edges to the root have infinite credit capacity, and transactions occur only between leaf nodes) and compute the steady-state transaction failure probability in a centralized system. We show that liquidity in star networks, complete graphs and Erd\"{o}s-R\'{e}nyi networks is comparable to that in equivalent centralized currency systems; thus we do not lose much liquidity in return for their robustness and decentralized properties.
\end{abstract}

\section{Introduction}
One of the primary functions of money is as a medium of exchange. This function is facilitated by governments and central banks that issue currency and declare it to be legal tender, \ie, the state promises to accept that currency back as payment for goods and services. Thus currency represents an obligation issued by the state, and when used to make a payment, results in a transfer of the state's obligation from the payer to the payee. A decision to accept payment in a currency is therefore a decision to trust the issuer of the currency to fulfill its obligations.

The modern banking system is a centralized currency infrastructure. The central bank sits at the root of the tree. It issues currency, government notes, etc. Retail banks, individuals and businesses form the leaves of this tree. Since they cannot print their own currency, they trust the central bank for an infinite amount of money (in theory, at least). Many non-governmental organizations also  issue their own currencies, known as \emph{scrip}. Examples of such scrip systems in the online world include Facebook credits, Yootles \cite{yootles}, P2P systems where such currency is used to solve the free-rider problem \cite{karma, aperjis:p2p-exchange},  and virtual worlds such as Everquest and Second Life that have their own currencies. A challenge with scrip systems is to determine how much currency to issue in order to ensure liquidity without devaluing the currency \cite{kash:opt-scrip}.

We study an alternate model of credit, credit networks, introduced independently by DeFigueiredo and Barr \cite{trustdavis}, by Ghosh et al. \cite{ghosh-et-al:mech-design-trust-networks}, and by Karlan et. al. \cite{mobius-szeidl:collateral}. We generalize the model based on the insight that payments can be made in any currency as long as it is trusted by the payee. So nodes in our model act as banks and print their own currency and trust each other for certain amounts of each other's currency. This allows them to transact with each other without the need for a common currency: as long as there is a chain of sufficient credit from the payee to the payer, arbitrary payments can be made by passing obligations between trusting agents in their respective currencies. We study the question of long-term liquidity (\ie, capacity to route payments) in credit networks when nodes repeatedly transact with each other.

{\bf Illustrative Example.} Consider a credit network with three nodes $u, v$ and $w$ with edge $(u,v)$ having credit limit $c_1$ (\ie, $u$ trusts $v$ for up to $c_1$ units of $v$'s currency) and edge $(v,w)$ having a credit limit $c_2$. If $w$ wants to purchase a product or service from $u$ worth $p$ units of $w$'s currency, where $p \le \min\{c_1, c_2\}$, the transaction proceeds by $w$ issuing an IOU to $v$ worth $p$ units of $w$'s currency and $v$ issuing an IOU to $u$ worth $p$ units of $v$'s currency. However, if $p > \min\{c_1, c_2\}$ the transaction fails. As a result of a successful transaction, the edges $(u,v)$ and $(v,w)$ now have capacities $c_1 - p$ and $c_2 - p$ respectively (since $v$ and $w$ have depleted part of their credit line). But this transaction also results in the creation of edges $(v,u)$ and $(w,v)$ each having capacity $p$, since now $v$ can trust $u$ for up to $p$ units of its own currency, and $w$ can trust $v$ for up to $p$ units of its own currency. Thus routing payments in credit networks is identical to routing residual flows in general flow networks.
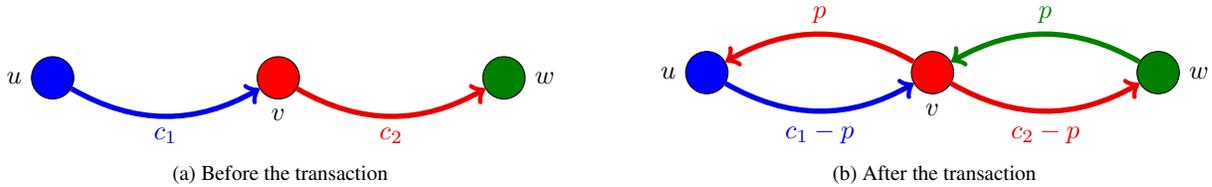
\begin{figure}[h]
\centering
\subfloat[Before the transaction]{
\begin{tikzpicture}
\tikzstyle{vertex}=[draw, circle,minimum size=16pt,inner sep=0pt]
\draw (0,0) node[vertex, fill=blue, label=left:$u$] (u) {};
\draw (3,0) node[vertex, fill=red, label=below:$v$] (v) {};
\draw (6,0) node[vertex, fill=green!50!black, label=right:$w$] (w) {};

\draw (u) edge[->, line width=2pt, bend right, color=blue!90!black] node[label=below:$c_1$] {} (v);
\draw (v) edge[->, line width=2pt, bend right, color=red!90!black] node[label=below:$c_2$] {} (w);
\end{tikzpicture}
} \hspace{1cm}
\subfloat[After the transaction]{
\begin{tikzpicture}
\tikzstyle{vertex}=[draw, circle,minimum size=16pt,inner sep=0pt]
\draw (0,0) node[vertex, fill=blue, label=left:$u$] (u) {};
\draw (3,0) node[vertex, fill=red, label=below:$v$] (v) {};
\draw (6,0) node[vertex, fill=green!50!black, label=right:$w$] (w) {};

\draw (u) edge[->, line width=2pt, bend right, color=blue!90!black] node[label=below:$c_1-p$] {} (v);
\draw (v) edge[->, line width=2pt, bend right, color=red!90!black] node[label=below:$c_2-p$] {} (w);
\draw (v) edge[->, line width=2pt, bend right, color=red!90!black] node[label=above:$p$] {} (u);
\draw (w) edge[->, line width=2pt, bend right, color=green!50!black] node[label= above:$p$] {} (v);
\end{tikzpicture}
} 
\caption{Illustrative Example}
\label{fig:illust-example}
\end{figure}

The papers that introduced this model used it as a generalization of budget constraints in auctions \cite{ghosh-et-al:mech-design-trust-networks} and as a mechanism for providing social collateral/insurance for transactions between untrusting agents in a network \cite{trustdavis, mobius-szeidl:collateral}. In addition to these papers, Resnick and Sami \cite{resnick-sami:transitive-trust} study how to update credit values along edges over time while ensuring robustness to Sybil attacks. They show that there is a social cost to allowing transactions via intermediaries---any protocol for updating credit values that is Sybil-proof in a certain sense can lead to a reduction in overall credit balances in expectation. Previous work on this model used the term ``trust network" for it; we prefer the term ``credit network" to distinguish this model from all the previous ones that also model trust/reputation in networks in settings unrelated to ours such as web search \cite{page-brin:pagerank, gyongyi:trustrank} and studying people's opinions of each other \cite{guha:trust, leskovec:signed-networks}.

Credit networks have a number of interesting and useful properties. First, the loss incurred by nodes in the network is \emph{bounded}. Consider a scenario where a potential attacker, say node $x$, injects nodes into this network. The attacker's nodes can extend each other credit lines worth arbitrary amount of money, but in order to transact with one of the honest players the attacker $x$ must dupe one of them, say $u$, into extending him a line of credit. If $u$ extends $x$ a credit line for $c$ units, $x$ can transact with honest nodes in the network for up to $c$ units by giving $u$ IOUs worth $c$ units  in its currency. Once the transaction is complete, $x$  can then disappear leaving $u$ with unredeemable IOUs worth $c$ units, and $u$ incurs a loss equal to $c$ units in $x$'s  currency. This demonstrates that the total loss incurred by the nodes in the network is bounded by the total credit extended by honest nodes to the attacker's nodes. In particular, the loss incurred by the nodes in the network does not depend on the number of nodes the attacker injects. Second, since a node in the network only accepts IOUs from other nodes to whom it has extended a line of credit, the loss incurred by honest players in the network is \emph{localized}. In other words, the only nodes that incur a loss are the ones that extended credit to the attacker nodes. In this sense, the system is fair; it protects a node from losses resulting from another node's lack of judgement. Finally, routing payments in credit networks is efficient since it only requires a max-flow computation \cite{trustdavis, ghosh-et-al:mech-design-trust-networks}.

These properties make this model useful not only for monetary transactions, but also in any setting where there is a need to model trust between nodes in a network. It is particularly well-suited for transactions in exchange economies such as P2P networks where it can be used to improve inefficiencies resulting from asynchronous demand and bilateral trading \cite{liu:trading}. It has been used as part of a system called Ostra \cite{mislove:ostra} designed to thwart spam in email. Ostra is similar in spirit to Gmail's Priority Inbox. This model can also be used in settings such as combating spam in viral marketing over social networks and packet routing in mobile ad-hoc networks. There is a large body of work in economics and sociology on social capital and favor exchanges in networks \cite{jackson:favor-exchange}. This model provides a rigorous way of not only keeping track of favors owed to and by each individual in a network, but also a way to exchange favors, via trusted intermediaries, between individuals that do not know each other directly.

However, in order for the model to be of practical use it should be able to support repeated transactions between nodes over a long period of time. This motivates the following question, which we formulate and study in this paper: if the network is sufficiently well-connected and has sufficient credit to begin with, can we sustain transactions in perpetuity without additional injection of credit? How does liquidity depend upon network topology and transaction rates between nodes, and how does it compare with the centralized model described above?

We show positive results for all of the above questions. We show that the failure probability of transactions is independent of the path used to route flow between nodes. For symmetric transaction rates, we show analytically and via simulations that the failure probability for complete graphs, Erd\"{o}s-R\'{e}nyi graphs, and preferential attachment graphs goes to zero as the size, the density or the credit capacity of the network increases. Further, we characterize a centralized currency system as a special type of a star network (one where edges to the root have infinite credit capacity, and transactions occur only between leaf nodes) and compute its steady-state failure probability. We show how to construct a centralized system that is equivalent to a given credit network and that the steady-state failure probability in star networks and complete graphs is within a constant-factor of that in equivalent centralized currency systems. We show that the node bankruptcy probability in credit networks is $\T(1/h_G)$ whereas that in an equivalent centralized system is $\T(1/\bar{c}_G)$, where $h_G$ is the harmonic mean of the total credit capacities on edges incident upon nodes in $G$ and $\bar{c}_G$ is the corresponding arithmetic mean. These results are similar in spirit to price of anarchy bounds---in return for all the benefits resulting from the decentralized nature of the system, we do not lose much liquidity compared to a centralized model with a common currency.

We would like to qualify our results by stating that nodes in our model are not endowed with any rationality. We treat credit values on edges and transaction rates between nodes as being exogenously defined. We also ignore currency exchange ratios, \ie, we assume them to be unity for clarity of exposition; this does not qualitatively affect our results. These questions represent promising directions for further study of this model and we discuss them at the end of the paper along with other open problems.

The rest of the paper is organized as follows: Section \ref{sec:model} defines our model and describes our main results, Section \ref{sec:analysis} contains steady-state analysis of the transaction success probability in our model and its comparison with a centralized model of credit,  Section \ref{sec:sim} discusses results of simulating repeated transactions on two random graph families, Erd\"{o}s-R\'{e}nyi graphs and power-law graphs constructed using the  Barab\'{a}si-Albert preferential attachment model \cite{barabasi-albert:pa}, and finally Section \ref{sec:conclusion} has some concluding remarks along with a discussion of a number of interesting directions for further study of this model.

\section{Our Model \& Results}\label{sec:model}
A credit network $G  = (V,E)$ is a directed graph with $n$ nodes and $m$ edges. Nodes represent entities or agents. Edges represent pairwise credit limits between agents. An edge $(u,v)\in E$ has capacity $c_{uv} > 0$, which implies that $u$ has extended a credit line of $c_{uv}$ units to $v$ in $v$'s currency. If a node $s$ needs to pay $p$ units in her currency to node $t$ (for example, to buy a good that $t$ is selling), the payment can go through if the maximum credit flow from $t$ to $s$ is at least $p$ units. The payment will get routed through a chain of trusted nodes from $s$ to $t$. Note that payment flows in the opposite direction of credit, so a payment merely results in a redistribution of credit; the total credit in the network remains unchanged.

We study liquidity under the following simple model of repeated transactions: we assume that each edge $(u,v)$ in $G$ has total credit capacity $c_{uv} + c_{vu} = c$ where $c_{uv}, c_{vu}$ are integers. The transaction rate between nodes is given by the $n\times n$ matrix $\L = \{\l_{uv} : u, v\in V; \l_{uu} = 0\}$ such that $\sum_{u,v} \l_{uv} = 1$. At each time step, we pick a pair of nodes $(s,t)$ with probability $\l_{st}$ and try to route a unit payment along the shortest feasible path from $s$ to $t$. If such a path exists, we route the flow along that path and modify edge capacities along the path as described previously (\ie, for each edge $(u,v)$ on the path, increment $c_{uv}$ by 1 and decrement $c_{vu}$ by 1). The transaction fails if there is no path between $s$ and $t$.

We begin by noting that at any time step an edge between nodes $u$ and $v$ can be in one of $(c+1)$ states, having capacity $c$ from $u$ to $v$ and 0 from $v$ to $u$, having capacity $c-1$ from $u$ to $v$ and 1 from $v$ to $u$ and so on down to having 0 capacity from $u$ to $v$ and $c$ from $v$ to $u$. Each successful transaction results in a change of state for all edges along the shortest path from the source to the destination. Therefore, these repeated transactions between nodes in $G$ induce the following Markov chain $\calM(G, \L)$: each state of the Markov chain encapsulates the state of all the edges in the network. Since there are $m$ edges, each of which can be in one of $(c+1)$ states, the Markov chain has $(c+1)^m$ states. A successful transaction results in a transition between two states. Let $P$ be the transition matrix. The transition probability, $P(\calS_1, \calS_2)$ from state $\calS_1$ to $\calS_2$ is equal to $\l_{st}$ where routing a unit payment from $s$ to $t$ in state $\calS_1$ along the shortest feasible path results in the network being in state $\calS_2$. The probability $P(\calS, \calS)$ is the probability that the transaction is infeasible when the network is in state $\calS$. This is equal to the sum of all probabilities $\l_{st}$ such that the network in state $\calS$ has no path from $s$ to $t$.

A note on terminology: we will sometimes refer to an edge $(u,v)$ as being ``bidirectional" if both $c_{uv}$ and $c_{vu}$ are non-zero, and as ``unidirectional" if exactly one of $c_{uv}$ and $c_{vu}$ is zero.

We describe our main results below. Our first result obviates the need to route payments along the shortest feasible path.
\begin{uthm}
Let $(s_1, t_1), (s_2, t_2), \dotsc, (s_T, t_T)$ be the set of transactions of value $v_1, v_2, \dotsc, v_T$ respectively that succeed when the payment from $s_i$ to $t_i$ is routed along a path $\calP_i$. Then the same set of transactions succeed when the payment from $s_i$ to $t_i$ is routed along any other feasible path $\calP'_i$.
\end{uthm}

In order to prove this, we observe that routing flow along a directed cycle changes the state of the network but does not affect the total credit available to any node.  This motivates the following definition:
\begin{mydef}
Given a credit network $G$, let $\calS$ and $\calS'$ be two states of $G$. We say that $\calS'$ is \emph{cycle-reachable} from $\calS$ if the network can be transformed from state $\calS$ to state $\calS'$ by routing a sequence of payments along feasible cycles (\ie, from a node to itself along a feasible path).
\end{mydef}
Cycle-reachability defines a partition $\calC$ over the set of possible states of the network; each equivalence class in $\calC$ is the set of states that are cycle-reachable from each other. We show the following result about the steady-state distribution of $\calM$ over equivalence classes in $\calC$:
\begin{uthm}
Consider a Markov chain $\calM_{\calS_0}(G,\L)$ starting in state $\calS_0$ induced by a symmetric transaction rate matrix $\L$ over nodes in $G$. Let $\calC_{\calS_0} \subseteq \calC$ be the set of equivalence classes accessible from $\calS_0$ under the regime defined by $\L$. Then $\calM_{\calS_0}$ has a uniform steady-state distribution over $\calC_{\calS_0}$.
\end{uthm}

We use this result to bound the steady-state node bankruptcy probability in general graphs and transaction failure probability for trees, cycles and complete graphs. A node being bankrupt is a sufficient but not necessary condition for a transaction with that node as the payer to fail; the necessary and sufficient condition for a transaction to fail is that the network has a directed cut of zero capacity between the payer and the payee. Let $d_v$ be the total credit capacity of all edges incident upon node $v$ in a graph $G$. Let $h_G$ be the harmonic mean of the $d_v$s and $\bar{d}_G$ be the arithmetic mean. We show that if $\calM(G,\L)$ is ergodic and $\L$ is symmetric, the average steady-state bankruptcy probability of a node in a credit network  is $\T(1/h_G)$ whereas that in an equivalent centralized system is $\T(1/\bar{d}_G)$.  We show that if $\L$ is uniform, the steady-state transaction failure probability in star networks is $\Theta(1/c)$, the steady-state success probability for line graphs is $\Theta(c/n^2)$ and that for cycles is $\Theta(c/n)$. We also show that if $\L$ is symmetric, the steady-state transaction failure probability for complete graphs is $\Theta(1/nc)$. Our results imply that for ``thin" graphs such as lines and cycles, the steady-state failure probability goes to one with the size of the network, whereas for well-connected graphs such as star networks, complete graphs and Erd\"{o}s-R\'{e}nyi networks\footnotemark[1], the steady-state failure probability goes to zero with the size, the average node degree or the credit capacity in the network.
\footnotetext[1]{For a connected Erd\"{o}s-R\'{e}nyi network, $G_c(n,p)$, the node bankruptcy probability is $\T(1/(npc))$.}

 \begin{table}
\centering
\begin{tabular}{| c | c | c | c | c |}
\hline
Network Topology & Transaction Regime & Credit Network & Centralized System\\
\hline
\hline
Star & uniform  & $\Theta(1/c)$ & $\Theta(1/c)$\\
Line & uniform & 1 - $\Theta(c/n^2)$ & $\Theta(1/c)$\\
Cycle & uniform & 1 - $\Theta(c/n)$ & $\Theta(1/c)$\\
Complete & symmetric  & $\Theta(1/nc)$ & $\Theta(1/nc)$\\
Erd\"{o}s-R\'{e}nyi\footnotemark[2]: $G_c(n,p)$ & symmetric  & $\Theta(1/npc)$  & $\Theta(1/npc)$\\
Barab\'{a}si-Albert\footnotemark[3]: $G^{BA}_c(n,d)$& symmetric & $\T(1/dc)$ & $\T(1/dc)$\\
\hline
\end{tabular}
\caption{Failure probability comparison between credit networks and equivalent centralized currency systems}
\label{table:liquidity-comp}
\end{table}
\footnotetext[2]{We conjecture based on simulations and heuristic calculations that the transaction failure probability is $\T(1/(npc))$.}
\footnotetext[3]{We conjecture based on simulations that the failure probability is $\T(1/(dc))$.}
Further, we  characterize a centralized currency system as a special type of a tree network (one where edges to the root have infinite credit capacity, and transactions occur only between leaf nodes) and compute its steady-state failure probability. We show how to construct a centralized system that is equivalent to a given credit network and that the steady-state failure probability in star networks and complete graphs is within a constant-factor of that in equivalent centralized currency systems. 

Table~\ref{table:liquidity-comp} shows our main analytical results on liquidity for a number of network topologies and their comparison with equivalent centralized currency systems. 

In addition to these analytical results, we simulated repeated transactions on two random graph families: Erd\"{o}s-R\'{e}nyi graphs and power-law graphs generated using the Barab\'{a}si-Albert preferential attachment (BA) model \cite{barabasi-albert:pa}. The simulations showed that relatively small values of density and credit capacity were sufficient to attain a transaction success probability in excess of 0.9. We also demonstrate that for both topologies, if we hold the average node degree and credit capacity on edges constant, the size of the network (\ie, number of nodes) had no effect on the steady-state success probability. 

\section{Analysis}\label{sec:analysis}
We first analyze the combinatorial structure induced by this model of making payments through passing IOUs. An understanding of this combinatorial structure allows us to characterize the steady-state behavior of the Markov chain induced by repeated transactions for a number of network topologies. Finally, we compare the steady-state success probability in credit networks with various topologies to that in an equivalent centralized currency infrastructure.

\subsection{Combinatorial Structure}\label{sec:comb-struct}
A state of the network can be associated with a \emph{generalized score vector}, which characterizes the credit available to each node in that state.
\begin{mydef}
Given a labeled, directed graph $D$ over $n$ nodes with edges having associated capacities, a vector $V = \langle v_1, \dotsc, v_n\rangle \in \real_+^n$ is the \emph{generalized score vector} of $D$ if the total capacity on outgoing edges from $i$ in $D$ is $v_i$ for $i = 1, \dotsc, n$.
\end{mydef}
This is a generalization of the score vector defined by the outdegree of nodes in an orientation of a labeled multigraph \cite{kleitman:forests-scorevectors}. Note that two cycle-reachable states have the same generalized score vector, and in fact, generalized score vectors characterize classes of cycle-reachable states:
\begin{lem}[Generalization of Proposition 4.10 in \cite{gioan:cycle-cocycle}]\label{lem:equiv-class-score-vectors}
Given a credit network $G$, two states $\calS$ and $\calS'$ of $G$ are cycle-reachable if and only if they have the same generalized score vector.
\end{lem}
Moreover, two cycle-reachable states have exactly the same set of feasible transactions:
\begin{lem}\label{lem:feasibility}
For any equivalence class $C\in \calC$ of a given network, if a transaction $(s,t,v)$ (\ie, routing a payment of $v$ units from node $s$ to node $t$) is feasible in some state $S \in C$, it is feasible in all states $S' \in C$.
\end{lem}

Using the definition of cycle-reachability and the above observations, we show the following: 
\begin{thm}
Let $(s_1, t_1), (s_2, t_2), \dotsc, (s_T, t_T)$ be the set of transactions of value $v_1, v_2, \dotsc, v_T$ respectively that succeed when the payment from $s_i$ to $t_i$ is routed along a path $\calP_i$. Then the same set of transactions succeed when the payment from $s_i$ to $t_i$ is routed along any other feasible path $\calP'_i$.
\end{thm} 
\begin{proof}
Note that a failed transaction does not change the state of the network, so without loss of generality, we assume that the set of successful transactions $(s_i, t_i)$ occurred in successive timesteps $i =1$ to $T$. We prove the result by induction on $T$. The statement clearly holds for $T=1$.

Assume that the statement holds for $T = k$. Let the initial state of the network be $\calS_0$. Let $\calS_k$ be the state of the network when, starting from $\calS_0$, transactions $(s_i, t_i)$ for $i = 1, \dotsc, k$ were routed along path $\calP_i$. Let $\calS'_k$ be the corresponding state of the network after, again starting from $\calS_0$, the same set of transactions were routed along paths $\calP'_i, i = 1, \dotsc, k$. We will show that $\calS_k$ is cycle-reachable from $\calS'_k$.

A successful transaction $(s,t)$ only changes the credit extended to $s$ and $t$; the credit extended to any of the intermediate nodes remains unchanged. Since $\calS_k$ and $\calS'_k$ are the resulting states of the network after the same sequence of payments were successfully routed starting at the same state, the credit distribution (or equivalently, the generalized score vectors) in $\calS_k$ and $\calS'_k$ must be identical. Therefore, $\calS_k$ and $\calS'_k$ are cycle-reachable. Thus, from Lemma~\ref{lem:feasibility} if the transaction $(s_{k+1}, t_{k+1})$ is feasible in $\calS_k$ then it is also feasible in $\calS'_k$. This proves the result.
\end{proof}
This result shows that routing in credit networks has a \emph{path-independence property}: the choice of paths along which a sequence of payments are routed does not affect their feasibility.  Lemma~\ref{lem:equiv-class-score-vectors} and Lemma~\ref{lem:feasibility} immediately implies the following corollary:
\begin{cor}\label{cor:reciprocal}
If a transaction $(s,t,v)$ in some state in equivalence class $C_i$ results in a transition to a state in equivalence class $C_j$, then the reverse transaction $(t,s,v)$ from any state in $C_j$ will result in a transition to a state in $C_i$.
\end{cor}

Note that the above observations do not require the edge capacities or payment flows to be integral. However, when restricted to a setting where edge capacities are integers and we route unit flows between nodes, the states of the credit network become equivalent to orientations of a labeled multigraph (since an edge of capacity $c$ can be viewed as $c$ labeled edges of unit capacity each).    It is known that there is a bijection between score vectors of a labeled multigraph and forests of a labeled multigraph \cite{kleitman:forests-scorevectors}. In the rest of the paper, we assume that edge capacities are integers and we route unit flows between nodes. This allows us to use the equivalence between forests, score vectors and cycle-reachable equivalence classes to characterize the steady-state of the Markov chain induced by repeated transactions.

\subsection{Steady-state Analysis}\label{subsec:ss-analysis}
Using observations about the combinatorial structure of the model, we next prove that the Markov chain induced by a symmetric transaction rate matrix $\L$ has a uniform steady-state distribution over $\calC$.
\begin{thm}\label{thm:main-ss}
Consider a Markov chain $\calM_{\calS_0}(G, \L)$ starting in state $\calS_0$ induced by a symmetric transaction rate matrix $\L$ over nodes in $G$. Let $\calC_{\calS_0} \subseteq \calC$ be the set of equivalence classes accessible from $\calS_0$ under the regime defined by $\L$. Then $\calM_{\calS_0}$ has a uniform steady-state distribution over $\calC_{\calS_0}$.
\end{thm}
\begin{proof}
As a consequence of the above facts, we can represent transactions as resulting in transitions between equivalence classes in $\calC$ instead of transitions between states of $\calM(G,\L)$. Let $\calT_{ij}$ be the set of transactions $(s,t)$ that result in a transition from some state in $C_i$ to some state in $C_j$. Then, (overloading the symbol $P$) we define the transition matrix over equivalence classes in $\calC_{\calS_0}$ as
\[
P(C_i, C_j) = \sum_{(s,t) \in \calT_{ij}} \l_{st}
\]
Note that Corollary~\ref{cor:reciprocal} implies that $(s,t) \in \calT_{ij}$ if and only if $(t,s) \in \calT_{ji}$. Further since $\L$ is symmetric, for any equivalence classes $C_i, C_j \in \calC_{\calS_0}, P(C_i, C_j) = P(C_j, C_i)$. Since $P$ is a symmetric stochastic  matrix, the uniform distribution is stationary with respect to $P$.
\end{proof}	
This immediately gives the following corollary:
\begin{cor}\label{cor:unif-ss-dist}
If $\calM(G,\L)$ is an ergodic Markov chain induced by a symmetric transaction rate matrix $\L$, it has a uniform steady state distribution over $\calC$.
\end{cor}

Next we show a sufficient condition for $\calM$ to be ergodic.
\begin{lem}
If for all $u \ne v, \l_{uv} > 0$, then $\calM$ is ergodic.
\end{lem}
\begin{proof}
$\calM$ has a finite number of states. Since $\forall u, v, u \ne v, \l_{uv} > 0$, for all states $\calS$ of $\calM$, $P(\calS, \calS) < 1$. Therefore, $\calM$ does not have a sink state. It is irreducible since for any states $\calS$ and $\calS'$ such that the ``edit distance" between $\calS$ and $\calS'$ is $k$, the $k$-step transition probability from $\calS$ to $\calS'$, $P^k(\calS, \calS') > 0$. Finally, to see that $\calM$ is aperiodic, consider a state $\calS$ such that $P(\calS, \calS) > 0$. For this state, both $P^2(\calS, \calS)$ and $P^3(\calS, \calS)$ are non-zero. Therefore, state $i$ is aperiodic. Since $\calM$ is finite-state, irreducible and aperiodic, it is ergodic.
\end{proof}

Theorem \ref{thm:main-ss} and Corollary~\ref{cor:unif-ss-dist} reduce the problem of computing the steady-state transaction failure probability to a counting problem. We instantiate this theorem for various network topologies to characterize their steady-state distribution and infer from it the steady-state failure probability in those topologies. For general graphs, we can use this theorem to bound the bankruptcy probability of a node.

\begin{thm}\label{thm:bankruptcy-prob}
Let $\calM(G, \L)$ be an ergodic Markov chain induced by a symmetric transaction matrix $\L$ over nodes in $G$. Fix a vertex $v$.  The steady-state probability that $v$ goes bankrupt is at most $1/(d_v + 1)$, where $d_v = \sum_{(v,u)\in E} (c_{uv} + c_{vu})$ is total credit capacity of all edges incident upon  node $v$.
\end{thm}
\begin{proof}
From Corollary~\ref{cor:unif-ss-dist}, we know that $\calM(G, \L)$ has a uniform steady-state distribution over the cycle-reachable equivalence classes of $G$. So the steady-state probability that $v$ goes bankrupt is simply the fraction of equivalence classes in which $v$ is bankrupt.

We know that the total number of cycle-reachable equivalence classes of $G$ is the number of subforests over all $n$ nodes, whereas the number of equivalence classes in which $v$ is bankrupt is equal to the number of subforests on the remaining $n-1$ nodes \cite{kleitman:forests-scorevectors}.  Let $F(G)$ be the total number of subforests of $G$ and $F(G_{-v})$ be the number of subforests on nodes other than $v$. An edge $e = (u,v)$ having total credit capacity $c_e$ can be viewed as $c_e$ labeled edges each of unit capacity. So, each subforest on the remaining $n-1$ nodes can be extended to a subforest on all $n$ nodes by adding one of the $d_v$ edges incident on $v$ or adding $v$ without adding any edges.  Therefore $F(G) \ge (d_v+1)F(G_{-v})$. Thus, the steady-state probability that $v$ goes bankrupt, $F(G_{-v})/F(G)$, is at most $1/(d_v + 1)$.  
\end{proof}
Therefore, if $h_G$ is the harmonic mean of the $d_v$s in $G$, then the average node bankruptcy probability in $G$ is $\T(1/h_G)$.

\subsubsection{Trees}
The above results allow us to make the following observations about the steady-state distribution over trees. We denote by $T_{n;c}$ a tree with $n$ nodes whose each edge has total credit capacity $c$.
\begin{thm}\label{thm:tree-unif}
If $\calM(T_{n;c}, \L)$ is an ergodic Markov chain induced by a symmetric transaction rate matrix $\L$ on a tree $T_{n;c}$, then $\calM(T_{n;c}, \L)$ has a uniform steady-state distribution over its states.
\end{thm}
\begin{proof}
The result follows directly from Corollary~\ref{cor:unif-ss-dist} and the fact that, since trees have no cycles, each equivalence class in $\calC$ is a singleton.
\end{proof}
Since $\calM(T_{n;c}, \L)$ has a uniform steady-state distribution over its states,  the steady-state probability that a transaction over a path of length $l$ will succeed is $[c/(c+1)]^l$; this follows from the fact that each edge along the path should have non-zero capacity in the direction of the transaction. Therefore, the steady-state success probability, $\prob_s$, for $T_{n;c}$ is given by
\[
\prob_s(T_{n;c}) = \expect_l \left(\frac{c}{c+1}\right)^l
\]
Using this expression for the steady-state success probability, we can show that if $\L$ is uniform (\ie, $\l_{st} = 1/n(n-1)$ for all $s, t$), the steady-state success probability for tree networks has a lower bound of $\T(c/n^2)$ and an upper bound of $\T(c/(c+1))$. The lower bound is attained for line networks and the upper bound for star networks.

\subsubsection{Cycles}
Next we derive the steady-state success probability for cycle graphs.  Let $G^\circ_{n;c}$ be a cycle graph of $n > 2$ nodes where each edge has a total credit capacity of $c$. Let  $\calM(G^\circ_{n;c}, \L)$ be an ergodic Markov chain induced by a symmetric transaction matrix $\L$. Then $\calM(G^\circ_{n;c}, \L)$ will have a uniform steady-state distribution over the cycle-reachable equivalence classes of $G^\circ_{n;c}$. We use this fact to show the following result:
\begin{thm}\label{thm:cycles}
If $\calM(G^\circ_{n;c}, \L)$ is an ergodic Markov chain induced by a symmetric transaction matrix $\L$ over nodes in $G^\circ_{n;c}$, then the steady-state transaction success probability, $\prob_s(G^\circ_{n;c})$, is given by
\[
\prob_s(G^\circ_{n;c}) = \expect_l \frac{r^l + r^{n-l} -2r^n}{1-r^n}
\]
where $r = c/(c+1)$ and $l$ is the number of edges between a pair of transacting nodes.
\end{thm}
The proof involves counting the total number of equivalence classes of $G^\circ_{n;c}$ and the number of classes which allow a transaction between a pair of nodes that have $l$ edges between them. See Appendix~\ref{app:cycles} for the proof.

\subsubsection{Complete Graphs}
Let $K_{n;c}$ be a complete graph over $n$ nodes such that each edge has total capacity $c$. From Theorem~\ref{thm:bankruptcy-prob}, we know that the bankruptcy probability of a node in $K_{n;c}$ is $\T(1/nc)$. Next we will show that the steady-state transaction failure probability between two nodes in $K_{n;c}$ is also $\T(1/nc)$. Note that a node being bankrupt is a sufficient but not necessary condition for a transaction with that node as the payer to fail.

\begin{thm}\label{thm:complete-graph}
Let $\calM(K_{n;c}, \L)$ be an ergodic Markov chain induced by a symmetric transaction matrix $\L$ over nodes in $K_{n;c}$. Then the steady-state transaction failure probability in $K_{n;c}$ is $\T(1/nc)$.
\end{thm}
See Appendix~\ref{app:complete-graph} for the proof. We will compare this probability with that in an equivalent centralized model to show that liquidity in credit networks with this topology is comparable to the centralized system.

\subsubsection{Erd\"{o}s-R\'{e}nyi Networks}\label{thm:gnp}
Here we prove that under a symmetric transaction regime, the steady-state probability that a node in a Erd\"{o}s-R\'{e}nyi network, $G_c(n,p)$, will go bankrupt is $\T(1/(npc))$. Note that $G_c(n,p)$ is a random graph over $n$ nodes where every edge is present with probability $p$ and has capacity $c$ if present.

\begin{thm}
Consider a $G_c(n,p)$ graph where $p > 8(\ln cn)/n$ (just slightly more than the connectivity threshold of $2\ln n/n$). Let $\calM(G_c(n,p), \L)$ be an ergodic Markov chain induced by a symmetric transaction matrix $\L$ over nodes in $G_c(n,p)$. Fix a vertex $v$.  The steady-state probability that $v$ goes bankrupt is $\T(1/(npc))$.
\end{thm}
\begin{proof}
From Theorem~\ref{thm:bankruptcy-prob}, we know that the bankruptcy probability of $v$ is at most $1/(d_v + 1)$. We need to show that in expectation this term is $\T(1/(npc))$. By Chernoff bounds, $\prob[d_v \le npc/2] \le e^{-np/8} \le e^{-\ln cn} = 1/(cn)$, which bounds $\expect_G[1/(d_v+1)]$ as follows:
\begin{align*}
\expect_G[1/(d_v+1)] & \le \prob[d_v < npc/2]\cdot 1 + \prob[d_v \ge npc/2]\frac{1}{npc/2}\\
& \le \frac{1}{cn} + \frac{2}{npc} = \T\left(\frac{1}{npc}\right)
\end{align*}
\end{proof}

We conjecture that the steady-state failure probability of transactions between nodes in a $G_c(n,p)$ network under a symmetric transaction regime is also $\T(1/(npc))$. This is based on heuristic calculations using the generating functions for complete graphs with $pc$ instead of $c$, as well as on simulation results (see Section~\ref{sec:var-cc}).
\begin{conj}\label{conj:gnp}
Let $G_c(n,p)$ be a connected credit network (\ie, $p > \ln n/n$) and let $\calM(G_c(n,p), \L)$ be an ergodic Markov chain induced by a symmetric transaction matrix $\L$ over nodes in $G_c(n,p)$. Then, the steady-state transaction failure probability in $G_c(n,p)$ is $\T(1/npc)$.
\end{conj}
As with complete graphs, we will also compute this probability in an equivalent centralized model and argue that this conjecture, if true, would imply that the failure probability in $G_c(n,p)$ credit networks is within a constant factor of the equivalent centralized model.

\subsubsection{Barab\'{a}si-Albert Preferential Attachment Networks}
We state a conjecture regarding the steady-state transaction failure probability in power-law graphs constructed using the Barab\'{a}si-Albert preferential attachment model. The Barab\'{a}si-Albert model is a evolving random graph model where each arriving node creates edges to existing nodes with probability proportional to their degrees. It is parameterized by the number of nodes in the network, $n$, and the number of edges $d$ that each new node creates.
\begin{conj}\label{conj:ba}
Let $G^{BA}_c(n, d)$ be a random graph over $n$ nodes created using the Barab\'{a}si-Albert preferential attachment model where each arriving node creates $d$ edges and each edge has total credit capacity $c$. Let $\calM(G^{BA}_c(n, d), \L)$ be an ergodic Markov chain induced by a symmetric transaction matrix $\L$ over nodes in $G^{BA}_c(n,d)$. Then, the steady-state  transaction failure probability in $G^{BA}_c(n, d)$ is $\T(1/dc)$.
\end{conj}
Note that we conjecture the transaction failure probability to be independent of network size $n$. This conjecture is supported by simulation results (see Section~\ref{sec:sim}).

\subsection{Centralized Currency Infrastructure}
Since we introduced credit networks as an alternative to a centralized currency infrastructure, a natural question arises:  how does the steady-state failure probability in a credit network compare with that in a centralized infrastructure? Next we address this question.

As described previously, a centralized currency infrastructure can be modeled as a tree with a bank at its root and individuals as leaves. However, this setup is different from a star network in that the root does not participate in any transactions; all transactions occur between leaf nodes, \ie, along paths of length two. Let the centralized infrastructure be represented by a network $\calG^*$ with a root node $r$ representing the bank, and $n$ leaf nodes representing individuals. Let $c_{ru} < \infty$ be the credit extended by bank $r$ to node $u$. Note that since the bank issues a common currency a leaf node trusts the bank for an infinite sum of money, \ie, $c_{ur} = \infty$. A successful transaction of unit value leads to an increase by one of the payee's credit and a decrease by one of the payer's credit; the total credit extended to leaf nodes in the system stays constant. We can prove that if the total credit extended to leaf nodes is $C$, then the steady-state transaction failure probability in a centralized currency system is $(n-1)/(C+n-1)$.
\begin{thm}\label{thm:central-ss}
Let $\calM(\calG^*, \L)$ be an ergodic Markov chain induced by a symmetric transaction matrix $\L$ over a centralized payment network $\calG^*$ with $n$ leaf nodes. Let $C$ be the total credit extended to the leaf nodes. Then the steady-state transaction failure probability in $\calM$ is  $(n-1)/(C+n-1)$.
\end{thm}
\begin{proof}
Since the number of leaf nodes is $n$ and the total credit extended to leaf nodes is $C$, the Markov chain has $\binom{C + n -1}{n-1}$ states: this is the number of ways of distributing $C$ objects among $n$ people. In the centralized system, unlike credit networks, a transaction fails if and only if the payer in the transaction, say node $u$, is bankrupt, \ie, $c_{ru} = 0$. Therefore, the transaction failure probability is equal to the probability that a node is bankrupt. The number of states in which a given node is bankrupt is given by $\binom{C + n - 2}{n-2}$: this is the number of ways of distributing the total credit among the remaining $n-1$ leaf nodes. From Corollary~\ref{cor:unif-ss-dist}, we know that $\calM$ has a uniform steady-state distribution over its states. Therefore, the steady-state probability that a node is bankrupt is given by
\[
\frac{\binom{C + n - 2}{n-2}}{\binom{C + n - 1}{n-1}}  = \frac{n-1}{C+n-1} 
\]
\end{proof}
So if $\bar{c} = C/n$ is the average credit extended to a node in $\calG^*$, then the node bankruptcy probability of a node (as well as the transaction failure probability) is $\T(1/\bar{c})$. Note that since the Markov chain has a uniform distribution over states of the network, the steady-state transaction failure probability is independent of the initial state (\ie, the initial credit of each node in the network).

\subsubsection{Equivalence between Credit Networks and Centralized Currency Infrastructure}
In order to be able to compare liquidity in the two models, we need to define the notion of equivalence between the two models. Given a credit network $G = (V,E)$ with $n$ nodes, we construct an equivalent centralized currency infrastructure $\calG(G)$ as follows: $\calG$ has nodes $V \cup \{r\}$ (where $r$ is the root node) and a bidirectional edge between each node $u$ and the root $r$. For each node $u$ in $G$, the total credit extended to it by the bank in the centralized model, $c_{ru}$, is given by $c_{ru} = \sum_{v\in V} c_{vu}$ where $c_{vu}$ is the initial credit extended by $v$ to $u$ before the nodes start transacting. And finally, for all nodes $u$, we set $c_{ur} = \infty$. This ensures that each node has access to the same amount of credit in both models. 

\subsubsection{Liquidity Comparison}
The procedure to convert a credit network in to an equivalent centralized currency system allows us to compare liquidity in credit networks with that in an equivalent currency system. For the credit network topologies that we analyzed in Section~\ref{subsec:ss-analysis}, we state the steady-state failure probability for their equivalent centralized systems under a symmetric, ergodic transaction regime by invoking Theorem~\ref{thm:central-ss}. 

\paragraph{Cycle Graph} For a cycle graph with $n$ nodes where each edge has capacity $c$, the total capacity is $nc$. So, the steady-state failure probability in the equivalent centralized system is $(n-1)/(nc+n-1) = \T(1/c)$. On the other hand, in a uniform transaction regime, credit networks on cycle graphs have a steady-state failure probability of $1 - \T(c/n)$ (see Appendix~\ref{app:cycles}). However, for graphs with high expansion, the failure probability in credit networks is within a constant factor of that in equivalent centralized systems (see below).

\paragraph{Complete Graph} For a complete graph on $n$ nodes where the total credit capacity between any pair of nodes is $c$, the total credit capacity in the system is $Nc$, where $N = \binom{n}{2}$. So, the steady-state failure probability in the equivalent centralized system is $(n-1)/(Nc+n-1) = \T(1/nc)$, the same as that in the original credit network (Theorem~\ref{thm:complete-graph}). 

\paragraph{Erd\"{o}s-R\'{e}nyi Graph} For an Erd\"{o}s-R\'{e}nyi network, $G_c(n,p)$, the expected total credit capacity in the network is $Npc$, where $N = \binom{n}{2}$. Therefore, the steady-state failure probability in the equivalent centralized system is $(n-1)/(Npc+n-1) = \T(1/npc)$. Thus, if Conjecture~\ref{conj:gnp} is true, the failure probability in $G_c(n,p)$ credit networks would be within a constant factor of that in the equivalent centralized system.

\paragraph{Barab\'{a}si-Albert Graph} For a Barab\'{a}si-Albert preferential attachment network, $G^{BA}_c(n, d)$, the total credit capacity in the network is $ndc$ (\ie, $nd$ edges each having total credit $c$). Therefore, the steady-state failure probability in the equivalent centralized system is $(n-1)/(ndc+n-1) = \T(1/dc)$. Like  Erd\"{o}s-R\'{e}nyi networks, if the conjecture about steady-state failure probability for  $G^{BA}_c(n, d)$ credit networks (Conjecture~\ref{conj:ba}) is true, then, credit networks with this topology too will have a steady-state failure probability within a constant factor of that in an equivalent centralized system.

\section{Simulations}\label{sec:sim}
Next we present the results of simulating repeated transactions on credit networks from two well-studied families of random graphs: $G_c(n,p)$ (this is a Erd\"{o}s-R\'{e}nyi random graph with edge orientation chosen randomly and each edge assigned capacity $c$) and power-law graphs generated using the Barab\'{a}si-Albert preferential attachment (BA) model \cite{barabasi-albert:pa}. The Barab\'{a}si-Albert model is a evolving random graph model where each new node creates $d$ edges to existing nodes with probability proportional to their degrees. After constructing the graph, we assigned edge orientations randomly and set the capacity of each edge to $c$. The goal of the simulations was to understand how the steady-state success probability in these networks   depends on key network parameters and to determine what values of network parameters yield a ``high enough" success probability. In this sense, the qualitative results obtained from running simulations complement the asymptotic analysis described in the previous section. For $G_c(n,p)$ the parameters of interest are $n, p$ and $c$. However, since varying either $n$ or $p$ changes the density of the graph (\ie, average node degree), we also ran simulations where we changed both $n$ and $p$ such that $np$ was held constant. For BA graphs, the parameters of interest are the number of nodes ($n$), the number of edges each node creates ($d$), and credit capacity ($c$) of each edge.

\begin{figure}[h]
\centering
\includegraphics[width=80mm]{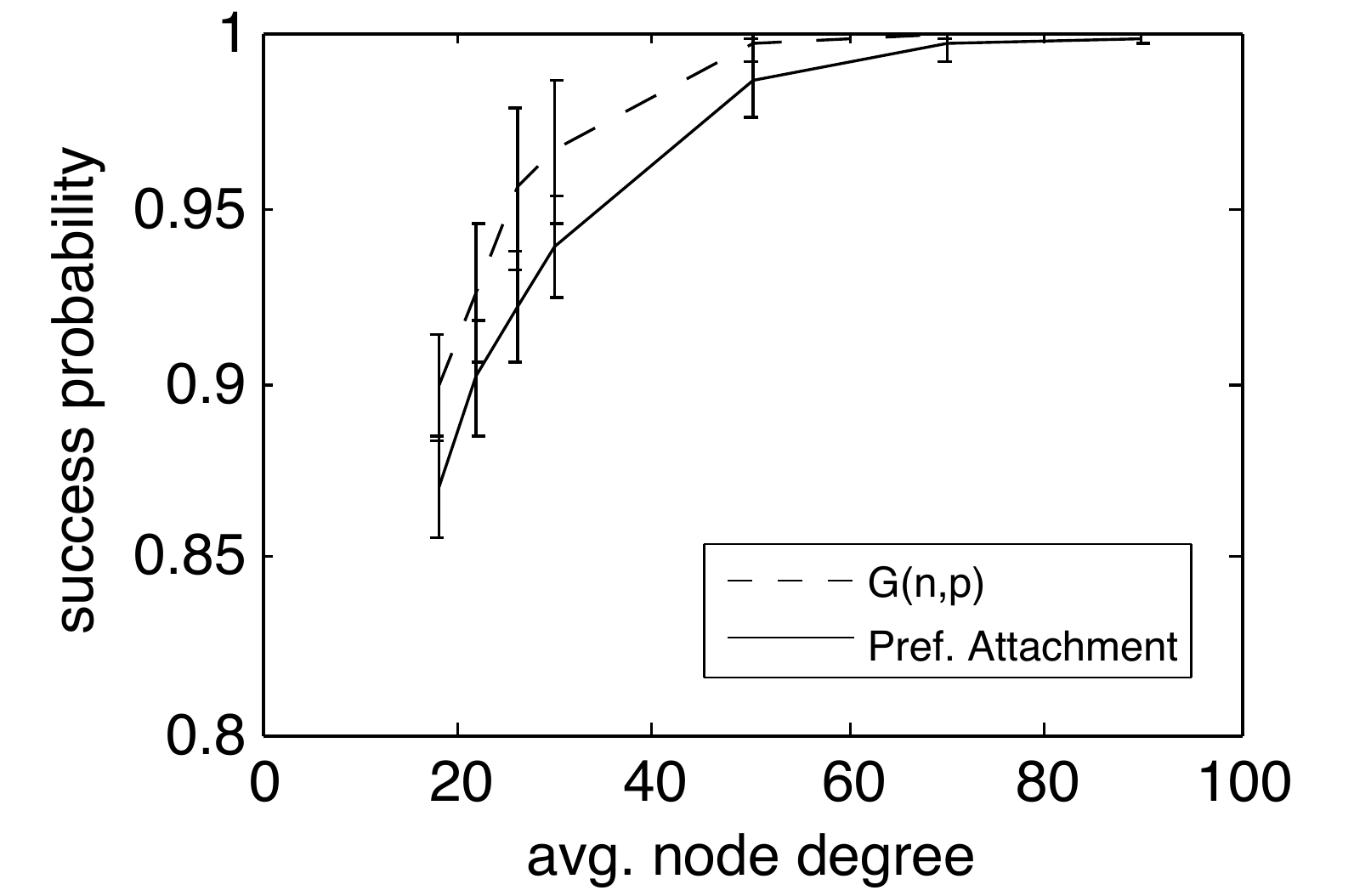}
\caption{Effect of varying graph density (average node degree) on the steady-state success probability in $G_c(n,p)$ and $G^{BA}_c(n, d)$ graphs with n = 200 and c = 1.}
\label{fig:gnp-pa-p}
\end{figure}

A simulation run consisted of constructing a network with a given set of parameters and doing repeated transactions on it. At each time step, we chose a node pair $(s,t)$ with uniform probability and tried to route a unit payment from $s$ to $t$ via the shortest feasible path. If there was a path from $s$ to $t$, we routed the payment and modified edge capacities along the path. Otherwise, we counted the transaction as a failure. We repeated this until the success-rate (i.e fraction of transactions that succeed) in two consecutive time windows was within $\ve$ of each other. We used a window size of 1000 time steps and $\ve = 0.002$. When the process converged, we measured the success rate at convergence, \ie, the total number of successful transactions divided by the total number of time steps. We called this the steady-state success probability for that set of network parameters. For each combination of network parameter values, we ran simulations on 100 graphs constructed from the same distribution and computed the average and the standard deviation in success-rate over the 100 runs. In order to understand how the success probability depends upon various network parameters, we repeated the set of 100 runs for different values of one of the parameters, keeping the rest constant.

\subsection{Effect of Variation in Network Density}

In order to study the effect of density, we fixed $n = 200$ and $c = 1$ for both graph families. For $G_c(n,p)$ we varied $p$ from 0.18 to 0.45 and for BA graphs we varied $d$ from 18 to 45. The average node degree of PA graphs is equal to $2d$ whereas for $G_c(n,p)$ the expected node degree is given by $(n-1)p$. We picked $d$ corresponding to each $p$ such that $2d \approx (n-1)p$ so that we can compare the plots for the two graph families. 

Figure \ref{fig:gnp-pa-p} shows the success probability as a function of average node degree for $G_c(n,p)$ and BA graphs when $np$ is sufficiently larger than $\ln n$ so that the $G_c(n,p)$ graph is connected. The success rate for both graph families is concave, non-decreasing in $p$. Note that for networks of 200 nodes with edge capacity of 1 and an average node degree as low as 25, the success probability for both topologies is in excess of 0.9. Moreover, the success probability quickly reaches a point of diminishing returns where making the network denser does not help significantly. This is useful since it means that, in practice, the network does not need to be very dense in order to route payments with a sufficiently high probability.

\begin{figure}[h]
\centering
\includegraphics[width=80mm]{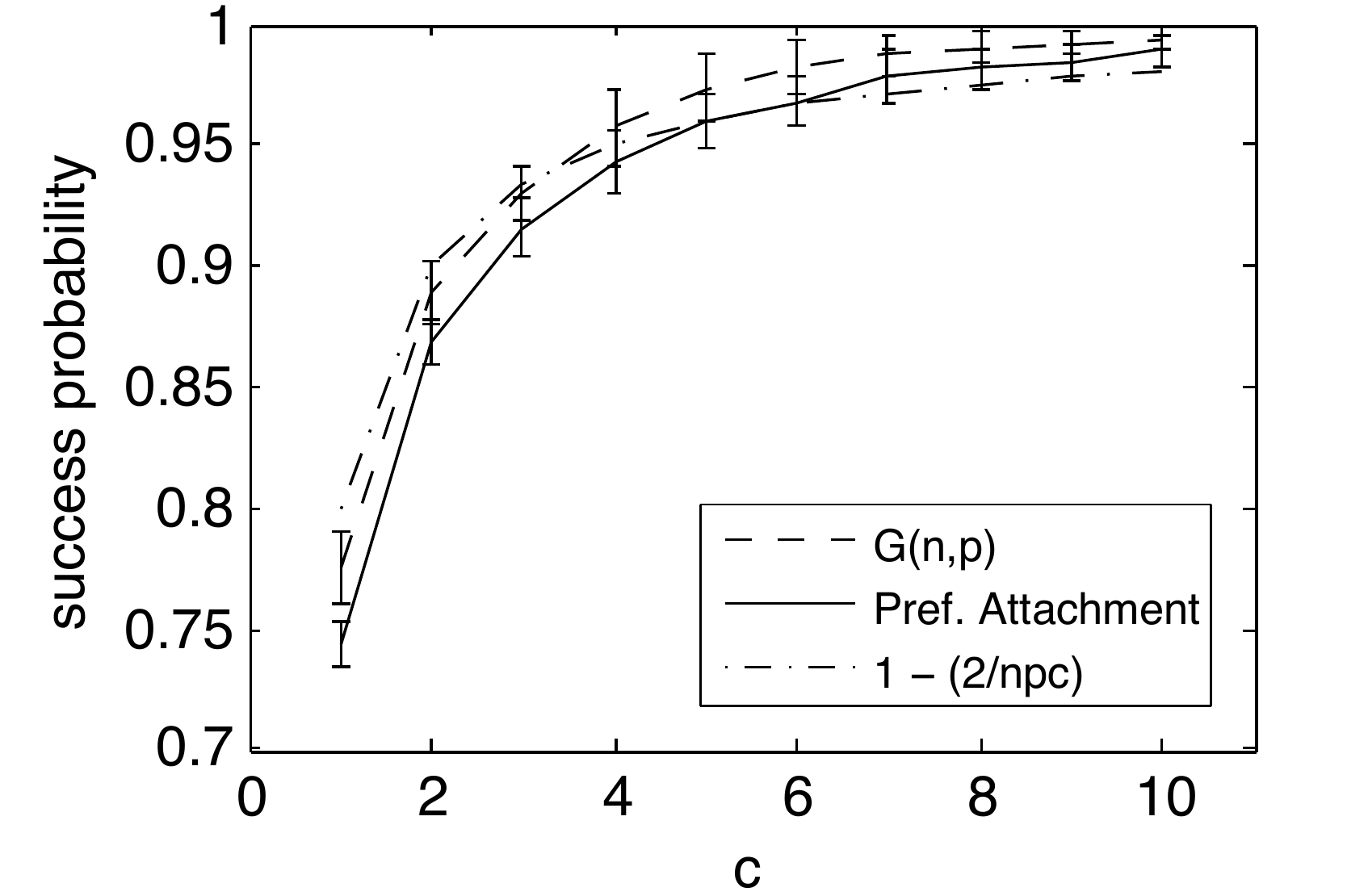}
\caption{Effect of varying credit capacity on the steady-state success probability in $G_c(n,p)$ and $G^{BA}_c(n, d)$ graphs with n = 100 and avg. node degree = 10.}
\label{fig:gnp-pa-c}
\end{figure}

\subsection{Effect of Variation in Credit Capacity}\label{sec:var-cc}

In order to study the effect of varying credit capacity in the network, we fixed $n = 100$ for both graphs. We set $p = 0.10$ for $G_c(n,p)$ and $d = 5$ for BA graphs, thereby ensuring that their densities were equal. This choice of $p$ and $d$ also ensured that all the graphs we constructed were connected. We varied $c$ from 1 to 10. Figure \ref{fig:gnp-pa-c} shows the  mean and standard deviation in the steady-state success-rate for both graph families as we vary $c$. In addition to the two plots, we also plot $1 - (2/npc)$ as a function of $c$. That plot tracks those for $G_c(n,p)$ as well as for BA fairly closely and is further evidence in support of our conjecture that the steady-state failure probability in $G_c(n,p)$ graphs is $\T(1/npc)$ and in $G^{BA}_c(n,d)$ is $\T(1/dc)$ and consequently, credit networks with these topologies are within a constant-factor of the equivalent centralized currency system in terms of liquidity.

\subsection{Effect of Variation in Network Size}
\begin{figure}[h]
\centering
\includegraphics[width=80mm]{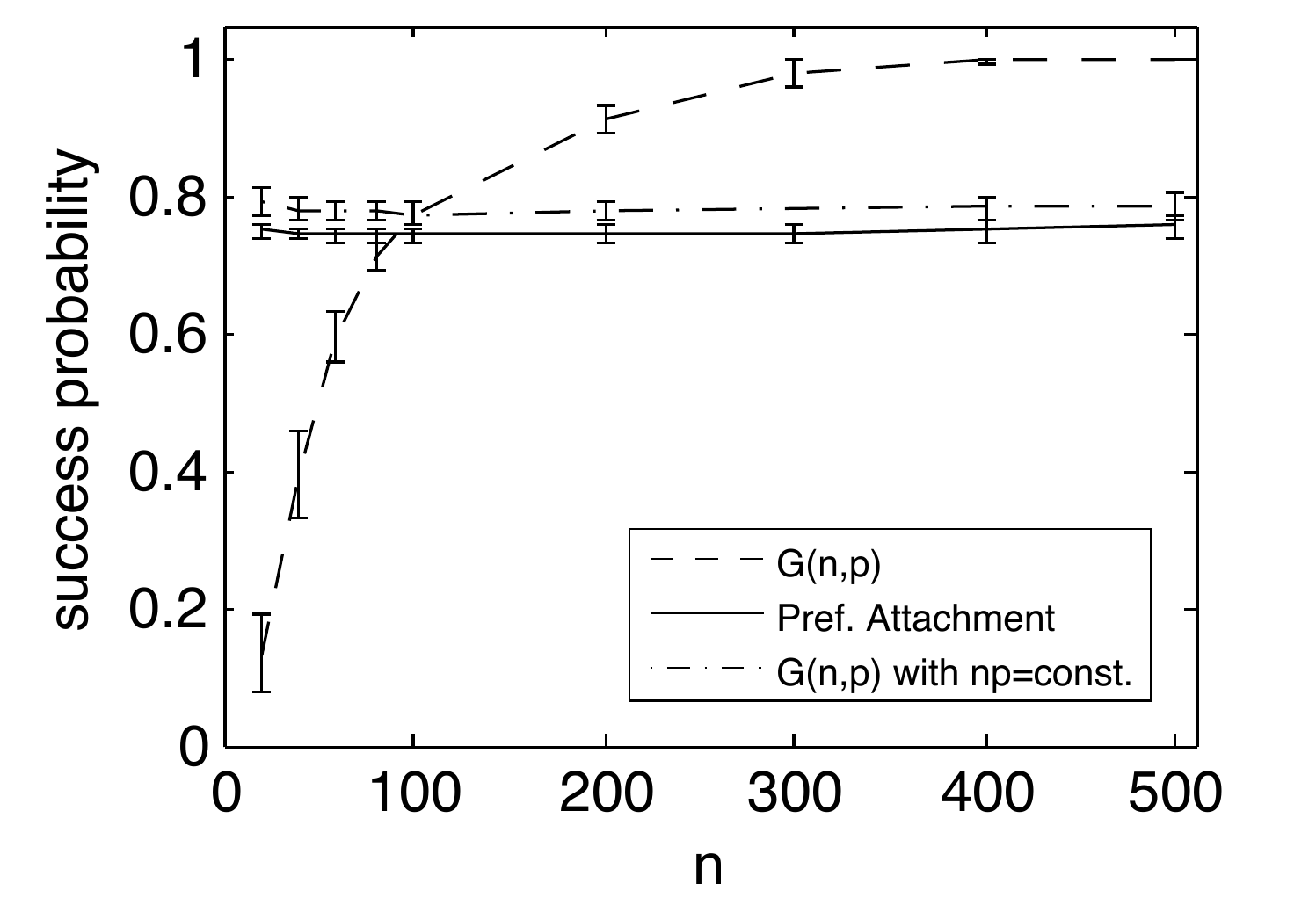}
\caption{Effect of varying network size (number of nodes) on the steady-state success probability in $G_c(n,p)$ and $G^{BA}_c(n, d)$ graphs with c = 1 and avg. node degree = 10.}
\label{fig:gnp-pa-n}
\end{figure}

In order to study the effect of varying network size for PA graphs, we fixed $d = 5$ and $c = 1$ and varied $n$ from 20 to 500. However, the same approach does not work for $G_c(n,p)$ graphs since if we fix $p$ and $c$ and varied $n$, it would also vary the average node degree. So, we ran two sets of simulations for $G_c(n,p)$ graphs: one where we varied $n$ keeping $p = 0.10$ and $c = 1$ and the other where we fixed $c = 1$ and varied both $n$ and $p$  such that $np$ was kept constant at 10.

Figure \ref{fig:gnp-pa-n} shows the effect of varying network size on the success probability of both graph families. The noteworthy observation is that if the density of the network is kept constant, then network size has no effect on the success rate for either graph family; the success rate remains nearly constant at about 0.75 for BA graphs and about 0.78 for $G_c(n,p)$ graphs. The same pattern is also observed if $np$ is held constant at 20 instead of 10 (so it is not an artifact of that specific value of $np$). This is another observation in support of our conjecture about the steady-state failure probability in $G_c(n,p)$ and BA graphs. If, however, $np$ is not held constant, then success probability is again a concave non-decreasing  function of $n$ (much the same way as it is of $p$).

\section{Conclusion \& Open Problems}\label{sec:conclusion}
In this paper we generalized the credit network model by allowing every node to print its own currency, and formulated and studied the question of long-term liquidity in this network under a simple model of repeated transactions. Using the notion of cycle-reachability, we showed that routing payments in credit networks has a \emph{path-independence} property. We also showed that the Markov chain induced by a symmetric transaction regime has a uniform steady-state distribution over the cycle-reachable equivalence classes. We used this fact to derive the node bankruptcy probability in general graphs and the transaction failure probability in a number of network topologies. We showed that, except in cycles, these probabilities are within a constant factor of the corresponding values for an equivalent centralized payment infrastructure. Our analysis and simulations show that for a number of well-known graph families, the steady-state failure probability under reasonable transaction regimes is comparable to that in equivalent centralized currency systems. Thus, in return for the robustness and decentralized properties, this model does not lose much liquidity compared to a centralized model. 

\subsection{Open Problems related to Liquidity}
Our results on liquidity can be extended in a number of ways. One question that arises from our analysis is whether credit network topologies with a high expansion have comparable liquidity with equivalent centralized systems. In general, can the steady state success probability in general networks be expressed in terms of some key network properties such as density or expansion? Another direction is to understand how success-rate and rate of convergence vary under various transaction regimes. A natural regime to study would be where the probability that a node is a payer (payee) is proportional to the node's in-degree (out-degree). We also leave unanswered the question of mixing times of the Markov chain:  our simulations on $G_c(n,p)$ and BA graphs converged relatively quickly but we do not have analytical results on the mixing time of the Markov chain and how it depends on credit capacity and network topology. One can also ask how fault-tolerant credit networks are under various models of node failures. For example, if we remove $k$ nodes (and all their edges) from the network, how would liquidity change in the worst-case, and what network topologies are most robust against such failures? If node failures cascade through the network, how quickly would the network lose liquidity? Another question worth investigating is the effect of transaction routing costs on liquidity. It appears that a mathematical analysis of a transaction routing model where intermediate nodes charged a routing fee would require an entirely new approach since it would invalidate the cycle-reachability relation that forms the basis of our results.

\subsection{Further Understanding the Model}
This model presents a number of interesting directions for further inquiry. The first is the question of pricing of currencies. Let $r_{uv}$ be the exchange ratio for converting $u$'s currency into $v$'s (\ie, 1 unit of $u$'s currency is $r_{uv}$ units of $v$'s currency). The following two properties of these conversion ratios are necessary for the generalized credit network model to be well-defined:
\begin{itemize}
\item \emph{Conservation:} For any $S \subseteq V$, where $S = \{u_1, u_2, \dotsc, u_k\}, \prod r_{u_1u_2}r_{u_2u_3}\dotsc, r_{u_{k-1}u_k}r_{u_ku_1} = 1$. In words, this property means that currency is conserved along any cycle of payments in the network. 
\item \emph{Common Knowledge:} The conversion ratios are common knowledge for all nodes in $V$. This means that nodes along a path through which payment is being routed can compute the outflow given the inflow.
\end{itemize}
However, these properties do not guarantee that it is incentive-compatible for nodes to route payments. For example, it is natural to believe that  the currency of a node that lies on the path between a number of (source, destination) pairs should be more expensive compared to one that has, say, only one neighbor. More generally, can we endow the nodes in the network with a model of rationality that will govern how much credit nodes extend each other, how that changes over time, and how nodes negotiate currency exchange ratios as a function of transaction rates, success probabilities, network topology, etc.?

The second broad question is how to apply this model to various practical applications, some of which we mention in the introduction. Of particular interest are analytical results as well as empirical or simulation-based studies that evaluate the effectiveness of this model in settings such as IP routing, peer-to-peer systems, viral marketing, etc.\ compared to existing models.

Another potential direction is to consider credit networks as enabling a trust-based market where nodes buy and sell goods using trusted currency instead of a common currency. Then we can ask the same questions that have been asked of various traditional market models: Does the market clear? At what prices? Can we efficiently find them? How do equilibrium prices in this model compare with those in models that assume a common currency? Answers to these questions will provide insights into whether and when can we trade off the robustness properties and the decentralized nature of this model with the well-understood characteristics of traditional markets.

\section{Acknowledgements}
This research was supported in part by NSF award IIS-0904325. Part of the research was also sponsored by the Army Research Laboratory and was accomplished under Cooperative Agreement Number W911NF-09-2-0053. The views and conclusions contained in this document are those of the authors and should not be interpreted as representing the official policies, either expressed or implied, of the Army Research Laboratory or the U.S. Government.  The U.S. Government is authorized to reproduce and distribute reprints for Government purposes notwithstanding any copyright notation here on.

\bibliographystyle{alphaurl}
\bibliography{references}

\appendix

\section{Success Probability in Cycles}\label{app:cycles}
\begin{proof}[Proof of Theorem~\ref{thm:cycles}]
Since $\calM$ has a uniform steady-state distribution over the cycle-reachable equivalence classes, the steady-state success probability between a pair of nodes separated by $l$ edges in $G^\circ_{m;c}$ is simply the fraction equivalence classes which allow a transaction between that pair of nodes.

Note that each equivalence class can be characterized by a state in which at least one edge in $G^\circ_{c,m}$ has zero capacity in the counterclockwise direction (any state where all edges are bidirectional is cycle-reachable from one where at least one edge is unidirectional). Assume that the nodes are numbered $1, \dotsc, m$ and the edge between nodes $j$ and $j+1$ is numbered $j$. Let the $j$th edge be the first edge that has zero capacity in the counterclockwise direction. Then the first $j-1$ edges can each be in one of $c$ states, and the remaining $m-j$ edges can each be in one of $c+1$ states. Thus, the total number of states where one or more edges have zero capacity in the counterclockwise direction is given by:
\begin{align*}
\sum_{j=1}^{m} c^{j-1}(c+1)^{m-j} &= (c+1)^{m-1}\sum_{j=1}^{m} \left(\frac{c}{c+1}\right)^{j-1} \\
&= (c+1)^{m-1}\frac{1-r^m}{1-r}
\end{align*}
where $r = c/(c+1)$.
Now fix nodes $s$ and $t$ that are a distance $l$ apart in the clockwise direction from $s$ to $t$. Lets call the path of length $l$ between $s$ and $t$, $\calP_1$, and that of length $m-l$, $\calP_2$. A successful transaction between $s$ and $t$ along $\calP_1$ requires edges in $\calP_1$ to have capacity in the counterclockwise direction, whereas that along $\calP_2$ requires edges in $\calP_2$ to have capacity in the clockwise direction. If the $j$th edge lies in $\calP_1$ then there is no path of length $l$ between $s$ and $t$. On the other hand, a path of length $m-l$ between $s$ and $t$ requires each of the $m-l$ edges to be in one of $c$ states. Thus, the number of equivalence classes that only have a path of length $m-l$ between $s$ and $t$ is
\begin{align*}\label{eq:path-m-l}
\sum_{j=1}^{l} c^{j-1}(c+1)^{l-j}c^{m-l} &= (c+1)^{l-1}c^{m-l} \sum_{j=1}^{l}\left(\frac{c}{c+1}\right)^{j-1}  \\
&=  c^{m-l}(c+1)^{l-1}\frac{1-r^{l}}{1-r}
\end{align*}
where $r=c/(c+1)$.
However, if the $j$th edge is in $\calP_2$, it implies that all edges in $\calP_1$ have positive capacity in the counterclockwise direction. Thus, the number of equivalence classes that have a path of length $l$ between $s$ and $t$ is
\begin{align*}
\sum_{j=l+1}^{m} c^{j-1}(c+1)^{m-j} &= c^l(c+1)^{m-l-1} \sum_{j=1}^{m-l} \left(\frac{c}{c+1}\right)^{j-1} \\
&= c^l(c+1)^{m-l-1} \frac{1-r^{m-l}}{1-r}
\end{align*}
Note that the above equation includes equivalence classes that have capacity along both $\calP_1$ and $\calP_2$ as well as those that have capacity along only $\calP_1$.

Thus, the probability that a transaction between nodes separated by $l$ edges will succeed is given by
\begin{align*}
\frac{c^l(c+1)^{m-l-1}(1-r^{m-l}) + c^{m-l}(c+1)^{l-1}(1-r^l)}{(c+1)^{m-1}(1-r^m)} \\
= \frac{r^l(1-r^{m-l}) + r^{m-l}(1-r^l)}{1-r^m} = \frac{r^l + r^{m-l} -2r^m}{1-r^m}
\end{align*}
Taking expectation over the length $l$ of the transaction path gives the result.
\end{proof}
As a corollary, we can show that the steady-state success probability in cycle graphs under a uniform transaction regime is $\T(c/n)$.
\begin{cor}
Let $G^\circ_{n;c}$ be a cycle graph of $n > 2$ nodes where each edge has a total credit capacity of $c$.  If $\calM(G^\circ_{n;c}, \L)$ is an ergodic Markov chain induced by a uniform transaction matrix $\L$ over nodes in $G^\circ_{n;c}$, then the steady-state transaction success probability, $\prob_s(G^\circ_{n;c})$, is $\T(c/n)$.
\end{cor}
\begin{proof}
 $\prob_s(G^\circ_{n;c})$ under a symmetric transaction regime is given by
\[
\prob_s(G^\circ) = \expect_l \frac{r^l + r^{n-l} - 2r^n}{1-r^n} \text{ (where $r = c/(c+1)$)}
\]
Therefore, under a uniform transaction regime, as $n\rightarrow \infty$
\begin{align*}
\lim_{n\rightarrow \infty} &\prob_s(G^{\circ})  \le \lim_{n\rightarrow \infty} \sum_{l=1}^{n/2} \frac{2n}{n(n-1)}\frac{r^l + r^{n-l} - 2r^n}{1-r^n}\\
 &\le \lim_{n\rightarrow \infty} \frac{2}{(n-1)(1-r^n)}\sum_{l=1}^{n/2} \left(r^l + r^{n-l}\right)\\
 &=  \lim_{n\rightarrow \infty} \frac{2}{(n-1)(1-r^n)}\left(\frac{r(1-r^{n/2}) + r^{n/2}(1-r^{n/2})}{1-r}\right)\\
 &= \frac{2r}{(n-1)(1-r)} = \T(c/n)
\end{align*}
\end{proof}

\section{Failure Probability in Complete Graphs}\label{app:complete-graph}
\begin{proof}[Proof of Theorem~\ref{thm:complete-graph}]
We know that there is a bijection between equivalence classes of a graph and labeled forests \cite{gioan:cycle-cocycle, kleitman:forests-scorevectors}, and the number of forests, $a_{n,c}$, in $K_{n;c}$ converges to \cite{kleitman:forests-scorevectors}
\begin{align}\label{complete-graph:eq1}
a_{n,c} \sim e^{1/2c}c^{n-1}n^{n-2}
\end{align}
where $f_n \sim g_n$ means $\lim_{n\rightarrow \infty} f_n/g_n = 1$. 
By Cayley's formula there are $n^{n-2}$ trees on $n$ labeled nodes, so $a_{n,1} \ge n^{n-2}$. Any forest has at most $n-1$ edges, so with $c$ labeled copies of each edge $c^{n-1}n^{n-2} \le a_{n,c} \le c^{n-1}a_{n,1}$.  Equation \eqref{complete-graph:eq1} implies that there exists a constant $n_0$ such that for all $n\ge n_0$, $a_{n,1} \le 2 e^{1/2}n^{n-2}$.  Combining these equations, for all $n \ge n_0$,
\[
c^{n-1}n^{n-2} \le a_{n,c} \le 2e^{1/2}c^{n-1}n^{n-2}
\]

Fix two nodes $u$ and $v$ in $K_{n;c}$, and consider a cut between $u$ and $v$ such that the partition containing $u$ has $k$ nodes and that containing $v$ has $n-k$ nodes. Assume that no payment can be made from $u$'s partition to $v$'s partition. Then, the number of equivalence classes containing this cut is at most $a_{k,c}a_{n-k,c}$. Summing over all cuts separating $u$ and $v$ in $K_{n;c}$, we upper bound the transaction failure probability between nodes $u$ and $v$ as
\begin{align*}
\frac{1}{a_{n,c}} \sum_{k=1}^{n-1}\binom{n-2}{k-1} a_{k,c}a_{n-k,c} &=  2\sum_{k=1}^{n_0-1} \binom{n-2}{k-1} \frac{a_{k,c}a_{n-k,c}}{a_{n,c}} + \sum_{k=n_0}^{n-n_0} \binom{n-2}{k-1} \frac{a_{k,c}a_{n-k,c}}{a_{n,c}}
\end{align*}

Bounding the first term,
\begin{align*}
\sum_{k=1}^{n_0-1} \binom{n-2}{k-1} \frac{a_{k,c}a_{n-k,c}}{a_{n,c}} &=  O\left(\sum_{k=1}^{n_0-1} \frac{n^{k-1}}{(k-1)!}\frac{(c^{k-1}a_{k,1})(c^{n-k-1}(n-k)^{n-k-2})}{c^{n-1}n^{n-2}}\right)\\
&= O\left( \frac{1}{nc} \sum_{k=1}^{n_0-1} \frac{a_{k,1}}{(k-1)!}\left(\frac{n-k}{n} \right)^{n-k-2} \right) = O\left(\frac{1}{nc}\right)
\end{align*}
using that $a_{k,1} = O(1)$ when $k \le n_0$. For the second term,
\begin{align*}
& \sum_{k=n_0}^{n-n_0} \binom{n-2}{k-1} \frac{a_{k,c}a_{n-k,c}}{a_{n,c}} \\
= & O\left(\sum_{k=n_0}^{n-n_0} \binom{n-2}{k-1}\frac{(c^{k-1}k^{k-2})(c^{n-k-1}(n-k)^{n-k-2})}{c^{n-1}n^{n-2}}\right) \\
= & O\left(\frac{1}{nc}\sum_{k=n_0}^{n-n_0} \binom{n-2}{k-1} \frac{k^{k-2}(n-k)^{n-k-2}}{n^{n-3}}\right) \\
= & O\left(\frac{1}{nc}\sum_{k=1}^{n-1} \frac{n!}{n(n-1)}\frac{k(n-k)}{k!(n-k)!}\frac{k^k(n-k)^{n-k}}{n^n}\frac{n^3}{k^2(n-k)^2}\right)
\end{align*}
\begin{equation}
\label{eq1} =  O\left(\frac{1}{nc}\sum_{k=1}^{n-1} \frac{n}{(n-k)k} \frac{n!}{n^n}\frac{k^k}{k!}\frac{(n-k)^{n-k}}{(n-k)!}\right)
\end{equation}

Instead of using the standard Stirling's approximation, we will use the more precise \emph{Stirling Series}:
\[
n! = \sqrt{2\pi n}\left(\frac{n}{e}\right)^n\left(1 + \frac1{12n} + \frac1{288n^2} - \frac{139}{51840n^3} \pm \cdots\right)
\]
Therefore, $n!$ can be written as
\[
n! = \sqrt{2\pi n}\left(\frac{n}{e}\right)^nc_n  \text{, where } 1 \le c_n \le 2 \text{ when } n \ge 1
\]
Substituting this into \eqref{eq1} and simplifying, we get
\[
\frac{1}{nc}\sum_{k=1}^{n-1} \frac{n}{(n-k)k} \frac{n!}{n^n}\frac{k^k}{k!}\frac{(n-k)^{n-k}}{(n-k)!}  =  O\left(\frac{1}{nc}\sum_{k=1}^{n-1} \left(\frac{n}{(n-k)k}\right)^{3/2}\right)
\]
Either $k$ or $n-k$ is at least $n/2$, so
\[
\sum_{k=1}^{n-1} \left(\frac{n}{(n-k)k}\right)^{3/2} \le 2\sum_{k=1}^{(n-1)/2} \left(\frac{2}{k}\right)^{3/2} = O(1)
\]
Summarizing,
\[
\frac{1}{a_{n,c}} \sum_{k=1}^{n-1}\binom{n-2}{k-1} a_{k,c}a_{n-k,c} = O\left(\frac{1}{nc}\right)
\]
Also, since the node bankruptcy probability in $K_{n;c}$ is $\T(1/nc)$, the transaction failure probability is $\Omega(1/nc)$. This completes the proof.
\end{proof}

\end{document}